\documentclass[prl,superscriptaddress,twocolumn]{revtex4}
\usepackage[utf8]{inputenc}
\usepackage{amsmath}
\usepackage{amsfonts}
\usepackage{amssymb}
\usepackage{graphicx}
\usepackage{braket}
\usepackage{dsfont}
\usepackage{framed}
\usepackage{float}
\usepackage{physics}
\usepackage{epstopdf}
\usepackage{physics}
\usepackage{multirow}
\usepackage{amsthm}
\usepackage{color}
\usepackage{tabularx}
\usepackage{algorithm}
\floatname{algorithm}{Protocol}
\usepackage{algorithmic}
\usepackage{mathtools}

\newenvironment{proof-sketch}{%
  \proof}{\endproof}

\begin{document}
\title{Authenticated teleportation with one-sided trust}
\author{Anupama Unnikrishnan}
\affiliation{Clarendon Laboratory, University of Oxford, Parks Road, Oxford OX1 3PU, United Kingdom}
\author{Damian Markham}
\affiliation{Laboratoire d'Informatique de Paris 6, CNRS, Sorbonne Universit\'e, 4 Place Jussieu, 75005 Paris}
%\date{\today}

\begin{abstract}
We introduce a protocol for authenticated teleportation, which can be proven secure even when the receiver does not trust their measurement devices, and is experimentally accessible. 
We use the technique of self-testing from the device-independent approach to quantum information, where we can characterise quantum states and measurements from the exhibited classical correlations alone. 
First, we derive self-testing bounds for the Bell state and Pauli $\sigma_X, \sigma_Z$ measurements, that are robust enough to be implemented in the lab. 
Then, we use these to determine a lower bound on the fidelity of an untested entangled state to be used for teleportation. 
Finally, we apply our results to propose a protocol for one-sided device-independent authenticated teleportation that is experimentally feasible in both the number of copies and fidelities required. 
This can be interpreted as a first practical authentication of a quantum channel, with additional one-sided device-independence.
\end{abstract}

\newtheorem{theorem}{Theorem}
\newtheorem{lemma}[theorem]{Lemma}
\newtheorem{corollary}[theorem]{Corollary}
\newtheorem{definition}{Definition}

\newtheorem{manualtheoreminner}{Theorem}
\newenvironment{manualtheorem}[1]{%
  \renewcommand\themanualtheoreminner{#1}%
  \manualtheoreminner
}{\endmanualtheoreminner}

\newtheorem{manuallemmainner}{Lemma}
\newenvironment{manuallemma}[1]{%
  \renewcommand\themanuallemmainner{#1}%
  \manuallemmainner
}{\endmanuallemmainner}

\maketitle

Quantum teleportation is well-established as a cornerstone of the field of quantum information, allowing the transfer of a qubit from one party to another using an entangled pair and a classical communication channel \cite{Bennett1993}. While interesting in its own right, it is also a key ingredient in many protocols, such as secret sharing \cite{Hillery1999, Cleve}, anonymous transmission \cite{Christandl2005} and multiparty computation \cite{Crepeau2002}, and is an important tool across quantum information.

From a cryptographic point of view, it is then vital to study the security of teleportation. We consider \textit{authenticated teleportation}, where we wish to verify that the teleportation has succeeded even when we do not trust the entangled pair being used. This authenticates its application as a quantum channel between the two parties. Previous schemes for authentication of a quantum channel, such as \cite{Barnuma}, rely on generating large entangled states or performing entangling measurements. In order to guarantee high security of such protocols, one would need levels of entanglement that are, in practice, unfeasible. 
We solve this problem and go one step further: allowing Alice and Bob to authenticate their quantum channel even when their devices may not be trusted. 

Device-independence has become a highly desirable feature of quantum communication and computation protocols, from its early applications to quantum key distribution \cite{Barrett2005, Acin2007, Vazirani2014} and quantum random number generation \cite{Colbeck2006, Pironio2010}. This approach addresses the situation where the untrusted components may have been obtained from, or be in the control of, an adversary. In the two party setting, one can consider two trust settings: one-sided device-independent (1sDI), where Alice trusts her device but Bob does not (or vice versa), and fully device-independent (DI), where neither party trusts their devices. 
One-sided trust should allow for a much simpler experimental implementation \cite{Branciard2012, Walk2016, Passaro2015}. This scenario is particularly relevant when one party is naturally trusted, for example, a trusted client but untrusted server, or simply if the channel and local measurement device are untrusted when one wishes to receive a resource (for example, a magic state for computation, or a particular entangled state for metrology).

In this work, we introduce a 1sDI protocol for authenticated teleportation. 
To achieve this, we present a protocol and security bounds for testing the entangled state which will be used for the teleportation, in a 1sDI way. The idea is essentially to request many entangled states, all but one of which are used as test runs, and the remaining one used for teleportation. Note that here, as in all authentication of quantum channels \cite{Barnuma}, one necessarily assumes an authenticated classical channel (or equivalently, a secure shared random key). Along the way, we will present some results for the full DI setting and discuss their application at the end. 
Finally, we argue that our protocol can be implemented with today's experimental capabilities. 

\textit{Outline.---} Building our protocol requires a few key ingredients. Firstly, we use the technique of \textit{self-testing}, by which untrusted states and measurements can be characterised, in a device-independent scenario, by the exhibited correlations alone. 
From its beginnings by Mayers and Yao \cite{Mayers2004}, self-testing results now encompass a variety of different states \cite{Mckaguea, McKague, Wu2014, Pal2014, Baccari2017, Coladangelo2017}, measurements \cite{Kaniewski2017, Bancal2015, Bowles2018}, and trust settings \cite{Gheorghiu, Supic}. For our application to teleportation, we focus on certifying the fidelity of the Bell state and the Pauli measurements.

By self-testing the Bell state, one can establish the closeness of an untrusted state with the Bell state, up to local isometry, from correlations such as the amount of violation of the Clauser-Horne-Shimony-Holt (CHSH) inequality \cite{Clauser1969}. 
While many self-testing results exist for the Bell state  \cite{Reichardt2013, Mckaguea, Gheorghiu}, only recently do we have techniques robust enough to be used in practice \cite{Yang, Bancal2015,Kaniewski2016}.
Self-testing a measurement certifies that, despite the possible presence of additional ancillae, the measurement operators act close to ideally on the Bell state, and trivially on the rest of the system.
There are several results in this direction \cite{Mckaguea, Gheorghiu, Supic, Kaniewski2017, Bancal2015}, however, no work so far has obtained a practically robust self-testing bound for the measurements. 

We start by expanding on the method introduced in \cite{Bancal2015}, known as the SWAP method, to derive full self-testing bounds for the Bell pair that incorporate both the state and Pauli $\sigma_X, \sigma_Z$ measurements, in both the 1sDI and DI settings. Our new bounds are practically robust. 
However, in this approach they are not immediately adapted to our requirements for several reasons. Firstly, they are based on the common self-testing assumptions of an infinite number of independent runs, throughout which the untrusted components behave identically (iid). 
In a realistic, adversarial scenario, we cannot rely on such assumptions, and so by adapting the methods of \cite{Gheorghiu, Hajdusek}, we remove these entirely. 
Secondly, the security statements that one achieves refer to the states that have been measured (hence cannot naively be used for ensuring the quality of the entangled pair used for teleportation). To address this, we incorporate a final, untested state in our analysis, and obtain a bound on the fidelity of teleportation using this state. 
This marks a departure from the usual self-testing works, which certify a state already consumed in the testing process. Recent work by Arnon-Friedman and Bancal \cite{Arnon-Friedman2017} also considered this in their non-iid entanglement certification scheme, however they focus on certifying distillable entanglement rather than quantifying the closeness of a state to the ideal. 

To demonstrate the experimental feasibility of our protocol, we give explicit values of the parameters we require to certify any fidelity of the teleportation, giving both the number of states and the violation of the inequality required. These results are within current experimental limitations \cite{Altepeter2005, Christensen2013, Kaiser2014, Saunders2010, Orieux2017, Abellan2018}.

\textit{Self-testing.---} 
Our first goal is to derive full, robust self-testing bounds for the Bell pair $(\ket{00} + \ket{11})/\sqrt{2}$ and Pauli measurements $\sigma_X, \sigma_Z$, thus adding to the work of \cite{Supic} in the 1sDI setting, and \cite{Bancal2015} in the DI setting. In order to do so, we use the same framework as defined there. 

In a device-independent scenario, one can only certify the state and measurements up to local isometry. Any local unitary transformation, or the presence of additional systems upon which the measurements act trivially, would result in the same correlations and so would not be detected. The dimension of the Hilbert space of the untrusted side is unrestricted, thus we can take the state to be pure and the measurements to be projective without loss of generality.

We start with considering the 1sDI setting. In this setting, it is relevant to look at steering correlations \cite{Wiseman2007}. Colloquially speaking, a quantum system is said to exhibit \textit{steering} if one party making local measurements can steer the qubit of the other party to a certain state. Analogous to the Bell inequality case, one can express this capacity via inequalities. For example, if Alice measures either one of the observables $A_0, A_1$, and Bob in one of $B_0, B_1$, one is said to observe steering if the following inequality \cite{Cavalcanti2009} is violated:
\begin{equation}
    \abs{ \langle A_0 B_0 \rangle + \langle A_1 B_1 \rangle } \leq \sqrt{2}. \label{EQN: steering inequality}
\end{equation} 
A violation of the inequality implies entanglement between the two parties. Further, a maximal violation of 2 can be demonstrated for the Bell pair under consideration if the observables are $A_0 = B_0 = \sigma_X, A_1 = B_1 = \sigma_Z$. 
We now give the definition of self-testing that we wish to make for the 1sDI setting, where the fidelity $F(\ket{\Theta}, \ket{\Omega}) = \abs{ \bra{\Theta}\ket{\Omega} }^2$. 

\begin{definition}
Let $\ket{\overline{\psi}}$ denote the ideal Bell state, and $\overline{M}_A, \overline{N}_B \in \{ \sigma_X, \sigma_Z \}$ be the ideal Pauli measurements of Alice and Bob respectively. Let $\ket{\psi}$ denote the untrusted shared state, and $N_B \in \{ X_B, Z_B \}$ be Bob's untrusted measurements. Given the parties observe a near-maximal violation $2-\epsilon$ of the steering inequality (\ref{EQN: steering inequality}), 
we have achieved robust self-testing if there exists some isometry $\Phi: \mathcal{H}_B \rightarrow \mathcal{H}_B \otimes \mathcal{H}_{B'}$ such that
\begin{align}
F \big(  \emph{tr}_B [\Phi(\ket{\psi}\bra{\psi})],  \ket{\overline{\psi}}  \big) & \geq 1 - f(\epsilon) \label{EQN: definition state self-testing} \\
F \big( \emph{tr}_B [ \Phi( \ket{\psi'}\bra{\psi'} ) ], \overline{M}_A \otimes \overline{N}_B \ket{\overline{\psi}} \big) & \geq 1 - f'(\epsilon) \label{EQN: definition measurement self testing}
\end{align}
where $\ket{\psi'} = \overline{M}_A \otimes N_B \ket{\psi}$ and $f(\epsilon)$, $f'(\epsilon)$ are known simple functions of $\epsilon$.
\label{def:selftesting}
\end{definition}

Note that, as is typical for self-testing statements, this bound effectively assumes an infinite number of iid copies of the state and measurements in order to approximate the violation and map it to the fidelity, and it is a statement on the state of the measured pairs.

We can define a similar statement for the DI setting, where now Alice's measurement $M_A$ is also untrusted. In this case the parties must observe a near-maximal violation of the CHSH inequality \cite{Clauser1969}, given by $\abs{ \langle A_0 B_0 \rangle + \langle A_1 B_0 \rangle + \langle A_0 B_1 \rangle - \langle A_1 B_1 \rangle } \leq 2$, and maximally violated at the Tsirelson bound of $2\sqrt{2}$ \cite{Cirelson1980}.

Our first aim is to determine forms of the functions $f(\epsilon), f'(\epsilon)$, for both the 1sDI and DI settings, that are practically robust i.e. give a non-trivial fidelity for experimentally observable violations of the inequality. For Equation (\ref{EQN: definition state self-testing}), such bounds for the state are already given in \cite{Supic, Bancal2015}, by semidefinite programming (SDP) and the SWAP isometry. We use their method to derive new bounds which, for the measurements, are significantly more robust than previous analytical results \cite{Supic, Gheorghiu, Mckaguea}.
Our results are given in Theorem \ref{th:selftestingbounds}.

\begin{theorem}
\begin{itemize}
\item[(a)] In the 1sDI setting, if the parties observe a $2 - \epsilon$ violation of the steering inequality (\ref{EQN: steering inequality}), then $
f(\epsilon)  = 1.26 \epsilon, 
f'(\epsilon)  = 3.10 \epsilon$.
\item[(b)] In the DI setting, if the parties observe a $2\sqrt{2} - \epsilon$ violation of the CHSH inequality \cite{Clauser1969}, then $f(\epsilon) = 1.19\epsilon, f'(\epsilon) = 3.70\epsilon$.
\end{itemize}
\label{th:selftestingbounds}
\end{theorem}

\begin{proof-sketch}  \renewcommand{\qedsymbol}{}
In the 1sDI case, let us denote $E_{\mathsf{b}|\mathsf{y}}$ as the projector associated with Bob measuring in setting $\mathsf{y}$ and getting outcome $\mathsf{b}$. 
Then, Alice's resulting conditional state is represented by the assemblage given by $\tau_{\mathsf{b}|\mathsf{y}} = \text{tr}_B (\mathds{1}_A \otimes E_{\mathsf{b}|\mathsf{y}} \ket{\psi} \bra{\psi})$, which is equal to the resulting state times its probability. 
In the DI case, we also consider Alice's projector $D_{\mathsf{a}|\mathsf{x}}$, such that the probability distribution is given by $p(\mathsf{a}, \mathsf{b} | \mathsf{x}, \mathsf{y}) = \text{tr}_{AB} (D_{\mathsf{a}|\mathsf{x}} \otimes E_{\mathsf{b}|\mathsf{y}} \ket{\psi} \bra{\psi})$.  

We then adopt the techniques of \cite{Bancal2015}, where by applying the SWAP isometry, we determine expressions for the fidelity measures in Equations (\ref{EQN: definition state self-testing}) and (\ref{EQN: definition measurement self testing}). We write our fidelity measures and inequality violations in terms of assemblage elements for the 1sDI setting, and expectation values of combinations of Alice and Bob's measurement operators for the DI setting (see Appendix A for more details).
Next, we introduce a positive semidefinite variable $\Gamma$, which is the moment matrix from the Navascu\'es-Pironio-Ac\'in (NPA) hierarchy \cite{Navascues2007, Navascues2008}, constructed such that now, the fidelity and inequality expressions can be written as a function of $\Gamma$. 
Since we wish to obtain a lower bound on the fidelity measures, we then use an SDP to find the minimum value of our expressions for Equations (\ref{EQN: definition state self-testing}) and (\ref{EQN: definition measurement self testing}) which are compatible with the amount of violation of the inequality. 

As we shift from self-testing the state to the measurements, the size of $\Gamma$ and the difficulty of the problem increases, particularly in constraining the structure of $\Gamma$. We solve this issue by automating the constraint generation step.
The full proof is given in Appendix A.
\end{proof-sketch}

From now on, we will denote the constant term by $\alpha$, such that the bounds in Theorem \ref{th:selftestingbounds} can be written as $f(\epsilon), f'(\epsilon) = \alpha \epsilon$, where $\alpha$ depends on the trust setting and whether we wish to use the self-testing bound for the state, or both state and measurements.

\textit{Authenticated teleportation.---}
We are now ready to build our protocol for authenticated teleportation with untrusted devices. As mentioned, we cannot directly apply our results from Theorem \ref{th:selftestingbounds}, as they are valid under the assumptions of an infinite number of iid rounds (appearing in the use of expectation values in the SDP). 
In any experiment, by doing a finite number of runs we can only get an estimation of the expectation values required for our inequality. Thus, if we want to use Theorem \ref{th:selftestingbounds} to propose a realistic verification protocol, we must also determine the number of runs required to approximate the expectation values to our desired precision. 
We consider the case of the iid assumption, and then remove this altogether for the fully adversarial setting. 
In this way, we define a scheme for 1sDI authenticated teleportation in Protocol \ref{alg:authtelep}.

The main idea of Protocol \ref{alg:authtelep} is as follows. The source is asked for many copies of the Bell pair. One copy is randomly chosen that will be used, and the rest are tested for steering using inequality (\ref{EQN: steering inequality}). If it passes, the remaining pair is used to teleport. If the source and devices behave as they should, the tests will always pass, and the teleported state is perfect. 
On the other hand, if the source is malicious, it cannot know which pair will be used and which pairs will be tested, and so if it supplies non-ideal states, it will sometimes fail the test. Further, the self-testing statements for the test mean that Bob's security holds even with untrusted devices on Bob's side.

Theorem \ref{th:fid} characterises the final untested state that is used as the quantum resource for teleportation, for both the case where we assume an iid source and the non-iid setting. Our proof for the non-iid case is based on techniques from \cite{Gheorghiu} to certify the fidelity of any randomly chosen state, as opposed to at least one state as in \cite{Bancal2015}.

\begin{algorithm}[t]
\caption{1sDI authenticated teleportation} 
\begin{algorithmic}[1]
\STATE 
Parameters $\epsilon, q, x$ are agreed, depending on the experimental limitations, and the fidelity of teleportation that the parties wish to certify. \\ \
\STATE The source is instructed to prepare
\begin{itemize}
\item $K = \lceil{ \frac{4 q^2 x}{\epsilon^2} \log{\frac{1}{\epsilon}} + 1 \rceil}$ Bell pairs in the iid setting, or 
\item $K = \lceil{  \frac{16 q^2 x}{\epsilon^2} \log{\frac{1}{\epsilon}} + 1 \rceil}$ Bell pairs in the non-iid setting, 
\end{itemize} 
and send the shares to Alice and Bob. ($\lceil{.\rceil}$ denotes the ceiling function.)  \\ \

\STATE Alice randomly chooses a pair $r$ to be used for the teleportation, and sends the value $r$ to Bob. \\ \

\STATE Alice randomly divides the set of remaining pairs into two subsets, $\mathbb{S}_0$ and $\mathbb{S}_1$, each of size $\frac{K-1}{2}$. \\ \

\STATE For each pair $i \in \mathbb{S}_t$: 
\begin{enumerate}
\item[(a)] Alice measures observable $A_t$ and gets outcome $a_i$.
\item[(b)]  She tells Bob to measure observable $B_t$ and he gets outcome $b_i$.
\item[(c)] Alice and Bob calculate their correlation for round $i$ as $\hat{C}_i = a_i b_i$. \\ \
\end{enumerate} 
\STATE Alice and Bob calculate their average correlation over all rounds. \\ \
\STATE If their average correlation has deviation $\epsilon$ from maximal violation of the steering inequality (\ref{EQN: steering inequality}), Alice uses pair $r$ to teleport the secret to Bob. 
\end{algorithmic}
\label{alg:authtelep}
\end{algorithm}

\begin{theorem}
The fidelity of the entangled state used for teleportation in step 7 of Protocol \ref{alg:authtelep} (up to local isometry) is bounded in the iid setting with probability at least $(1 - \epsilon^x)$ by
\begin{align}
F \geq 1 - \alpha \Big[ \frac{2\epsilon}{q} + \epsilon \Big],
\end{align}
and in the non-iid setting with probability at least $ (1-\epsilon^x)(1 - \sqrt{ \alpha [ \frac{2\epsilon}{q} + \frac{\epsilon}{2} + \frac{4q^2 x \epsilon \log{\frac{1}{\epsilon}} + 2 \epsilon^2}{8 q^2 x \log{\frac{1}{\epsilon}} +  \epsilon^2} ] })$ by
\begin{align}
F \geq 1 - \sqrt{ \alpha \Big[ \frac{2\epsilon}{q} + \frac{\epsilon}{2} + \frac{4q^2 x \epsilon \log{\frac{1}{\epsilon}} + 2 \epsilon^2}{8 q^2 x \log{\frac{1}{\epsilon}} +  \epsilon^2} \Big]  },
\end{align}
\label{th:fid}
where $\alpha = 1.26$.
\end{theorem}
\begin{proof-sketch}\renewcommand{\qedsymbol}{}
In Protocol \ref{alg:authtelep}, the final Bell pair is accepted in step 7 if, for all previous tested pairs, the measured average correlation was $\epsilon$-close to the correlation for an ideal pair. Our aim is to use this information to bound the fidelity of the final untested state to be used for teleportation. 

First, following the method of \cite{Gheorghiu}, we determine the closeness between the true correlation (i.e. the expectation value) and the ideal correlation, using this measured correlation. In the iid case, we do this by taking the measurement outcomes to be independent random variables, and then using the Chernoff-Hoeffding bound \cite{Chernoff1952, Hoeffding1963}. In the non-iid case, we instead use the martingale approach of Pironio \text{et al.} \cite{Pironio2010} and apply the Azuma-Hoeffding inequality \cite{Azuma1967, Hoeffding1963}. 

In our analysis, we also take into account the untested state $r$, which is straightforward if we assume an iid source. For the non-iid case, we do this by considering the maximum hypothetical deviation from the ideal correlation as in \cite{Hajdusek}. 

At this stage, we have an expression for the amount of violation of our steering inequality from the measured average correlations. Then, we apply our result from Theorem \ref{th:selftestingbounds} to bound the fidelity of the average state.  Our final result is a bound on any state prepared by the source, including state $r$ which is used for teleportation.

We introduce the parameters $q, x$ in order to tailor our protocol to possible experimental implementations, depending on the relative ease of generating many states or observing a high inequality violation. 
The full proof is given in Appendix B. 
\end{proof-sketch}

The bound on the fidelity of teleportation directly follows from this, as given in Corollary \ref{co:telepfid}.
\begin{corollary}
The fidelity of the teleported state (up to local isometry) in Protocol \ref{alg:authtelep} is lower-bounded by Theorem \ref{th:fid}. 
\label{co:telepfid}
\end{corollary}
\begin{proof}\renewcommand{\qedsymbol}{}
It is known that the fidelity of teleportation is at least as high as the fidelity of the entangled state used for teleportation \cite{Horodecki}, thus our bound on the entangled state also holds for the teleported state. 
\end{proof}

There is some flexibility in the protocol which would not affect our results. For example, we have written it such that Alice chooses the pair $r$, and essentially runs the protocol, but this can be easily changed to Bob, or a third party. Furthermore, as in \cite{Barnuma, Marin}, it is possible to make the protocol less interactive by replacing the communication between Alice and Bob with shared random strings, indicating which copy to use for the teleportation, and which measurement setting to test each copy with. 

In the way that we propose our protocol above, a quantum memory is needed to store the copies until they are tested or used. To override this experimentally difficult requirement, the source can generate Bell pairs on-the-fly, and the parties either test or use each pair depending on the chosen $r$. In this case, since we may not have checked for our desired inequality violation before the teleportation step (for example, if $r=1$ we immediately use the state before any testing), we run the risk of teleporting through a bad quantum channel. However, once we know our inequality violation at the end of the protocol, we can compute a bound on the fidelity of the teleportation from Theorem \ref{th:fid}. Alternatively, we could modify our protocol in the style of Pappa et al. \cite{Pappa2012}. Here, one uses the state with probability $2^{-K}$ in each of the $K$ rounds and tests otherwise, and the parties abort if any test fails. This would require an adaptation of the above proof. 

Finally, we note that the CHSH inequality can also be used to witness steering. The analysis of Protocol \ref{alg:authtelep} in this case, including the number of pairs required, follows from the proof of Theorem \ref{th:fid} by applying the relevant self-testing result, and is given in Appendix C.

\bigskip

\textit{Discussion and experimental feasibility.---}
First, we compare the security and resources required in our protocol alongside that of \cite{Barnuma}. While the protocol in \cite{Barnuma} has an exponential scaling with the security parameter, the size of the encoding, effectively the entanglement, increases with the desired security level. On the other hand, our protocol only uses multiple copies of the Bell pair.
As we outline below, due to the ease of generating entangled pairs (compared to high entanglement), our protocol is practically feasible. Furthermore, it is secure even in the case where devices are not trusted, not dealt with in \cite{Barnuma}.

Next, we discuss our self-testing results for the Pauli measurements, given as a fidelity measure. Using a steering inequality violation of at least 1.94 in the 1sDI case, or a CHSH violation of 2.78 in the DI case, one can certify a fidelity greater than 80\%. As such high Bell violations have been observed \cite{Kaiser2014, Abellan2018}, our bounds are sufficiently robust to be experimentally useful. 

To facilitate comparison with previous works, our bounds in Theorem \ref{th:selftestingbounds} are given in the same form as the literature in Table III (Appendix A). We see that we have improved upon all existing full self-testing bounds for the state and measurements in both the 1sDI and DI settings. Kaniewski's analytical approach \cite{Kaniewski2016} used a different isometry than SWAP and resulted in a better bound for the state in the DI setting. Any improved self-testing bound can directly be used to bound the fidelity of teleportation in our protocol using Corollary \ref{co:telepfid} and the parameter $\alpha$. 

We will now assess the experimental feasibility of Protocol \ref{alg:authtelep}.  
In order to improve upon teleportation by classical communication, we require the fidelity of teleportation to be greater than $\frac{2}{3}$ \cite{Popescu1994}. 
For example, 1sDI authenticated teleportation could be demonstrated, under the iid (non-iid) assumption, by observing a steering inequality violation of 1.75 (1.92) using $10^5$ ($10^8$) copies. 
If, instead, the CHSH inequality was used as the test, one would require a violation in the iid (non-iid) setting of 2.49 (2.73) with $10^6$ ($10^8$) copies.  (See Figure 2 in Appendix C for a plot of the full results). 
Such results could be demonstrated using Werner states $\rho(v) = v \ket{\overline{\psi}}\bra{\overline{\psi}} + (1-v) \frac{\mathds{1}}{4}$ with visibility $v \geq 0.88$ in the iid and $v \geq 0.96$ in the non-iid settings, which is well within experimental limitations. However, any implementation will also be subject to other imperfections and losses, such as in the photon detection efficiency. Using previous experimentally demonstrated violations of the steering inequality with two measurement settings \cite{Saunders2010, Orieux2017}, in the iid scenario one could certify a teleportation fidelity of at least 60\%. The relative abundance of experimental work on the CHSH inequality shows that the number of copies and violation (even in the non-iid setting) required for our scheme could be demonstrated in existing labs, for example \cite{Altepeter2005, Christensen2013, Abellan2018}.
Additionally, the number of copies required by our finite analysis is sufficiently high so as to provide a good enough estimate of the quantum distribution, meaning that we do not need to employ regularisation methods as in \cite{Lin2018}. 
Our scheme is thus the first authenticated teleportation protocol that is practical in its robustness, implementable with existing experimental setups, and further, tolerates untrusted devices.

We now comment on possible extensions of this work. Note that our protocol is a one-sided device-independent method of authenticating the quantum channel in teleportation. One can use our results to do a fully DI test of the Bell pair (the required parameters are in Appendix C), but in order to build a fully device-independent teleportation scheme, we must also consider self-testing of Alice's Bell state measurement. Recently, Renou et al. \cite{Renou2018} and Bancal et al. \cite{Bancal2018} have focused on this particular problem. While \cite{Bancal2018} gives practically robust bounds, they still assume iid and infinite runs, as is common in self-testing. Thus, we cannot directly apply these results to our protocol, and it remains as further work. 

We note that we have not considered the detection loophole in this work, and so have implicitly assumed sufficiently high detection efficiencies (as discussed in \cite{Pironio2009}, in a cryptographic setting we need not be concerned with other loopholes). An extension of our work that closes the detection loophole could take the no-detection events into account in the steering inequality, as was experimentally demonstrated in \cite{Wittmann2012, Smith2012, Bennet2012, Wollmann2016}. Although these reported violations are not yet high enough to demonstrate our scheme in a loophole-free manner, we expect that the development of 1sDI protocols will stimulate experimental efforts to achieve this in the near future. In fact, investigating other steering inequalities, possibly with more measurement settings, could prove useful in further optimising our protocol for experimental implementation. However, for the CHSH inequality, it is already possible to achieve high violations  while closing the detection loophole \cite{Christensen2013}. 

Finally, we emphasise that teleportation is not the only application of our work. Our results in Theorem \ref{th:fid} for the 1sDI setting, and in Appendix C for the DI setting, give a bound on the fidelity of the final untested state, which can then be used for a variety of other applications such as verified quantum computation. 

\bigskip

\noindent \textbf{Acknowledgements.} We thank Matty Hoban, Alexandru Gheorghiu, Eleni Diamanti, Simon Neves, Adeline Orieux and Valerio Scarani for fruitful discussions.
We acknowledge support of the EPSRC and the ANR through the ANR-17-CE24-0035 VanQute project.

%\bibliographystyle{ieeetr}
%\bibliography{reportbib}

\section{Appendix A: Proof of Theorem 1}
\label{sec: Appendix A}
\begin{manualtheorem}{1}
\begin{itemize}
\item[(a)] In the 1sDI setting, if the parties observe a $2 - \epsilon$ violation of the steering inequality, then $
f(\epsilon)  = 1.26 \epsilon, 
f'(\epsilon)  = 3.10 \epsilon$.
\item[(b)] In the DI setting, if the parties observe a $2\sqrt{2} - \epsilon$ violation of the CHSH inequality, then $f(\epsilon) = 1.19\epsilon, f'(\epsilon) = 3.70\epsilon$.
\end{itemize}
\label{th:selftestingbounds}
\end{manualtheorem}
\begin{proof}
Our proof is given in the following subsections.
$ $ \newline
\subsubsection{1. Fidelity expressions for 1sDI}
We first look at how to get an expression for the fidelity of the self-testing where we follow and extend the method of \cite{Supic}, using the SWAP isometry, which will later be fed into the SDP optimisation. Let Bob's untrusted device be a black box with measurement settings $\mathsf{y} \in \{ 0, 1 \}$ and outcomes $\mathsf{b} \in \{ 0, 1 \}$. The untrusted measurements made by Bob are then written in terms of the projectors $E_{\mathsf{b}|\mathsf{y}}$ as $X_B = 2E_{0|1} - \mathds{1}, Z_B = 2E_{0|0} - \mathds{1}$. 

The SWAP isometry $\Phi = \mathds{1}_A \otimes \Phi_B$ (Figure \ref{fig:isometry}), the standard in self-testing literature, performs the SWAP operation if the untrusted devices are operating correctly. It is composed of adding a trusted ancilla in a known state ($\ket{+}$) to Bob's subsystem, and then a unitary transformation which swaps part of the untrusted state onto the ancilla, resulting in a two-qubit state. 

The unitary transformation is written in such a way that it incorporates Bob's untrusted operators that we wish to test. Denoting Bob's system as $B$ and his ancilla system as $B'$, this transformation is given by $VHU$, where
\begin{align}
V & = \ket{0}\bra{0}_{B'} \otimes \mathds{1}_B  + \ket{1}\bra{1}_{B'} \otimes X_B, \\
H & = \ket{+}\bra{0}_{B'} + \ket{-}\bra{1}_{B'}, \\
U & = \ket{0}\bra{0}_{B'}  \otimes \mathds{1}_B + \ket{1}\bra{1}_{B'} \otimes Z_B.
\end{align}
The isometry then acts on the untrusted state as
\begin{align}
\Phi(\ket{\psi}) =  \mathds{1}_A \otimes (VHU)_{B' B} \ket{\psi}_{AB} \ket{+}_{B'}.
\end{align}
To self-test the state, we will determine a bound on the fidelity $F(\text{tr}_B [\Phi(\ket{\psi}\bra{\psi})], \ket{\overline{\psi}}) = \bra{\overline{\psi}} \text{tr}_B \big[ \Phi(\ket{\psi}\bra{\psi}) \big] \ket{\overline{\psi}}$.

In addition, we wish to self-test the untrusted measurement operators $N_B$.
To study this, we look at the action of the isometry given by
\begin{align}
\Phi(\mathds{1}_A \otimes N_B \ket{\psi}) = 
  \mathds{1}_A \otimes (VHU)_{B' B} \mathds{1}_A \otimes N_B \ket{\psi}_{AB} \ket{+}_{B'}.
\end{align}
In the ideal case, we have $
\mathds{1}_A \otimes \overline{N}_B \ket{\overline{\psi}}$. We will determine the closeness between these two expressions, and denoting $\ket{\psi'} = \overline{M}_A \otimes N_B \ket{\psi}$, we will compute a bound on the fidelity $F(\text{tr}_B [ \Phi(\ket{\psi'}\bra{\psi'}) ], \overline{M}_A \otimes \overline{N}_B \ket{\overline{\psi}})$. 

Our fidelity expressions will be in terms of assemblages on Alice's side, given by $\tau_{\mathsf{b}|\mathsf{y}} = \text{tr}_B (\mathds{1}_A \otimes E_{\mathsf{b}|\mathsf{y}} \ket{\psi} \bra{\psi})$. The full forms of these expressions are given in Table \ref{table:1sdifid}.

\begin{figure}[t]
\centering
\includegraphics[trim = 0mm 7cm 0mm 1cm, width=0.5\textwidth]{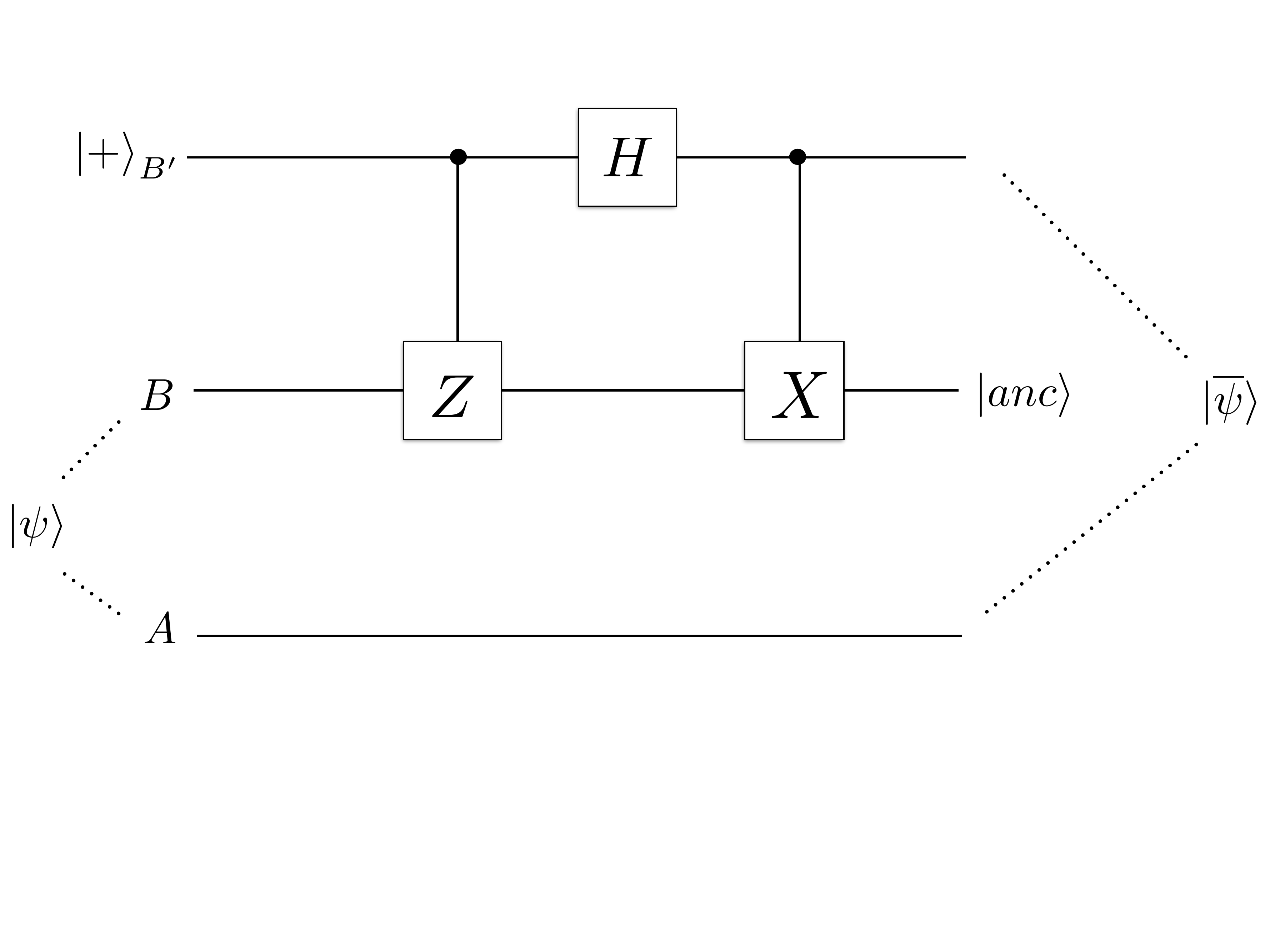}
\caption{The SWAP isometry applied to Bob's system for self-testing the state $\ket{\psi}$. \label{fig:isometry}}
\end{figure}

\begin{table*}[ht]
\centering
\begin{ruledtabular}
\begin{tabular}{   p{1.5cm}  p{14.5cm}  }
\multicolumn{1}{c}{\textit{To test}} & \multicolumn{1}{c}{\textit{Fidelity expression}} \\ \hline

 \begin{align*} \text{State}  \end{align*} & \begin{equation*} \begin{aligned}
 \frac{1}{2} \Big[ & \bra{0_A}  \tau_{0|0} \ket{0_A}  + \bra{0_A} (  2 \tau_{0|1, 0|0} - 2 \tau_{0|0, 0|1, 0|0} ) \ket{1_A}  \\
 & + \bra{1_A} (  2 \tau_{0|0, 0|1} - 2 \tau_{0|0, 0|1, 0|0}) \ket{0_A}    + \bra{1_A} (\rho_A - \tau_{0|0})  \ket{1_A} \Big] 
\end{aligned} \end{equation*} \\ 

 \begin{align*} Z_B \end{align*} & \begin{equation*} \begin{aligned}
 \frac{1}{2} \Big[ & \bra{0_A}  \tau_{0|0} \ket{0_A} + \bra{0_A} (  2 \tau_{0|1, 0|0} - 2 \tau_{0|0, 0|1, 0|0} ) \ket{1_A}  \\
 & + \bra{1_A} (  2 \tau_{0|0, 0|1} - 2 \tau_{0|0, 0|1, 0|0}) \ket{0_A}   + \bra{1_A} (\rho_A - \tau_{0|0})  \ket{1_A} \Big] 
\end{aligned} \end{equation*} \\ 

 \begin{align*} X_B \end{align*} & \begin{equation*} \begin{aligned} \frac{1}{2} \Big[ & \bra{0_A} ( \rho_A - \tau_{0|0} - 4 \tau_{0|1, 0|0, 0|1} + 2 \tau_{0|0, 0|1} + 2 \tau_{0|1, 0|0} ) \ket{0_A} \\
& +  \bra{0_A} (- 2 \tau_{0|0, 0|1} + 4 \tau_{0|0, 0|1, 0|0, 0|1} 
 + 4 \tau_{0|1, 0|0, 0|1, 0|0} + 4 \tau_{0|1, 0|0, 0|1}   - 2 \tau_{0|0, 0|1, 0|0} \\ 
&  - 8 \tau_{0|1, 0|0, 0|1, 0|0, 0|1}) \ket{1_A} 
 +  \bra{1_A} ( - 2 \tau_{0|1, 0|0}  + 4 \tau_{0|0, 0|1, 0|0, 0|1} + 4 \tau_{0|1, 0|0, 0|1, 0|0} + 4 \tau_{0|1, 0|0, 0|1}  \\
 & - 2 \tau_{0|0, 0|1, 0|0} - 8 \tau_{0|1, 0|0, 0|1, 0|0, 0|1})\ket{0_A}    +  \bra{1_A} ( 4 \tau_{0|1, 0|0, 0|1} - 2 \tau_{0|0, 0|1} - 2 \tau_{0|1, 0|0} + \tau_{0|0} ) \ket{1_A}
\Big] \end{aligned} \end{equation*} \\ 

\end{tabular}
\end{ruledtabular}
\caption{Fidelity expressions for the 1sDI setting. \label{table:1sdifid}}
\end{table*}

\subsubsection{2. Fidelity expressions for DI}

We use a similar framework for self-testing in the DI setting, following \cite{Bancal2015}. Since both Alice and Bob's devices are now untrusted, we apply the SWAP isometry to both parties' systems to extract a Bell pair. For ease of calculations, we now take the ideal state to be
\begin{align}
\ket{\overline{\psi}} = \frac{1}{\sqrt{2}} \begin{bmatrix} \cos{(\frac{\pi}{8})} \\[0.3em] \sin{(\frac{\pi}{8})} \\[0.3em] \sin{(\frac{\pi}{8})} \\[0.3em] -\cos{(\frac{\pi}{8})} \end{bmatrix},
\end{align}
which is equivalent to the Bell state up to local unitaries. This state maximally violates the CHSH inequality if Alice and Bob measure the $\sigma_Z, \sigma_X$ operators, which simplifies our analysis.
The isometry $\Phi = \Phi_A \otimes \Phi_B$ performs the same SWAP operation here, but we write it slightly differently to follow \cite{Bancal2015}. The ancilla on both parties' systems $A', B'$ is initialised in state $\ket{0}$, and the unitary transformation is now written as $VRV$, where $V$ is defined as before and $R$ is given by
\begin{align}
R = \mathds{1}_{A'} \otimes \Big( \frac{\mathds{1}_A + Z_A}{2} \Big) + \sigma_{X_{A'}} \otimes \Big( \frac{\mathds{1}_A - Z_A}{2} \Big)
\end{align}
on Alice's side, for example. For self-testing the state, we then have
\begin{align}
\Phi(\ket{\psi}) =  (VRV)_{A' A} \otimes (VRV)_{B' B} \ket{\psi}_{AB}  \otimes \ket{00}_{A' B'}.
\end{align}
The fidelity is given by $F(\text{tr}_{AB} [\Phi(\ket{\psi}\bra{\psi})], \ket{\overline{\psi}})  = \bra{\overline{\psi}} \text{tr}_{AB} \big[ \Phi(\ket{\psi} \bra{\psi} ) \big] \ket{\overline{\psi}}$.

For self-testing the measurements, we now denote Alice's untrusted measurement operator as $M_A$ and $\ket{\psi'} = M_A \otimes N_B \ket{\psi}$. After applying the isometry, we have
\begin{align}
\Phi( & M_A \otimes N_B \ket{\psi})  = \nonumber \\
&  (VRV)_{A' A} \otimes (VRV)_{B' B} M_A \otimes N_B \ket{\psi}_{AB} \ket{00}_{A'B'}.
\end{align}
The ideal operators acting on the ideal state gives $
\ket{\overline{\psi'}} = \overline{M}_A \otimes \overline{N}_B \ket{\overline{\psi}}$. 
The fidelity is then given by $
F(\text{tr}_{AB} [ \Phi(\ket{\psi'}\bra{\psi'})], \overline{M}_A \otimes \overline{N}_B \ket{\overline{\psi}})  = \bra{\overline{\psi'}} \text{tr}_{AB} [ \Phi(\ket{\psi'} \bra{\psi'}) ] \ket{\overline{\psi'}}$.

The full forms of these fidelity expressions are given in Table \ref{table:difid}.

\begin{table*}[ht]
\centering
\begin{ruledtabular}
\begin{tabular}{   p{1.5cm}  p{14.5cm}  }
\multicolumn{1}{c}{\textit{To test}} & \multicolumn{1}{c}{\textit{Fidelity expression}} \\ \hline

 \begin{align*} \text{State} \end{align*} & \begin{equation*} \begin{aligned} \frac{1}{2} \Big[ \frac{1}{2} & + \frac{1}{2\sqrt{2}} \langle Z_A Z_B \rangle + \frac{1}{4\sqrt{2}} \langle Z_A X_B \rangle 
+ \frac{1}{4\sqrt{2}} \langle X_A Z_B \rangle  - \frac{1}{8\sqrt{2}} \langle X_A X_B \rangle  - \frac{1}{8} \langle Z_A X_A Z_B X_B \rangle \\
&  - \frac{1}{8} \langle X_A Z_A X_B Z_B \rangle  + \frac{1}{8} \langle X_A Z_A Z_B X_B \rangle + \frac{1}{8} \langle Z_A X_A X_B Z_B \rangle + \frac{1}{8\sqrt{2}} \langle Z_A X_A Z_A X_B \rangle \\
& + \frac{1}{8\sqrt{2}} \langle X_A Z_B X_B Z_B \rangle  - \frac{1}{4\sqrt{2}} \langle Z_A X_A Z_A Z_B \rangle  - \frac{1}{4\sqrt{2}} \langle Z_A Z_B X_B Z_B \rangle  - \frac{1}{8\sqrt{2}} \langle Z_A X_A Z_A Z_B X_B Z_B \rangle \Big] \end{aligned} \end{equation*} \\ 

 \begin{align*} Z_A \otimes Z_B \end{align*} & \begin{equation*} \begin{aligned} \frac{1}{2} \Big[ \frac{1}{2} & + \frac{1}{2\sqrt{2}} \langle Z_A Z_B \rangle + \frac{1}{4\sqrt{2}} \langle Z_A X_B \rangle 
+ \frac{1}{4\sqrt{2}} \langle X_A Z_B \rangle  - \frac{1}{8\sqrt{2}} \langle X_A X_B \rangle  - \frac{1}{8} \langle Z_A X_A Z_B X_B \rangle \\
&  - \frac{1}{8} \langle X_A Z_A X_B Z_B \rangle   + \frac{1}{8} \langle X_A Z_A Z_B X_B \rangle + \frac{1}{8} \langle Z_A X_A X_B Z_B \rangle + \frac{1}{8\sqrt{2}} \langle Z_A X_A Z_A X_B \rangle \\
& + \frac{1}{8\sqrt{2}} \langle X_A Z_B X_B Z_B \rangle  - \frac{1}{4\sqrt{2}} \langle Z_A X_A Z_A Z_B \rangle - \frac{1}{4\sqrt{2}} \langle Z_A Z_B X_B Z_B \rangle  - \frac{1}{8\sqrt{2}} \langle Z_A X_A Z_A Z_B X_B Z_B \rangle \Big] \end{aligned} \end{equation*} \\ 

 \begin{align*} X_A \otimes X_B \end{align*} & \begin{equation*} \begin{aligned} \frac{1}{2} \Big[ \frac{1}{2} & - \frac{1}{8\sqrt{2}} \langle X_A X_B \rangle - \frac{1}{4\sqrt{2}} \langle X_A Z_A X_A X_B \rangle  - \frac{1}{4\sqrt{2}} \langle X_A X_B Z_B X_B \rangle  - \frac{1}{8} \langle Z_A X_A Z_B X_B \rangle \\
&  - \frac{1}{8} \langle X_A Z_A X_B Z_B \rangle  + \frac{1}{8} \langle X_A Z_A Z_B X_B \rangle + \frac{1}{8} \langle Z_A X_A X_B Z_B \rangle  + \frac{1}{2\sqrt{2}} \langle X_A Z_A X_A X_B Z_B X_B \rangle \\
& + \frac{1}{8\sqrt{2}} \langle X_A Z_A X_A Z_A X_A X_B \rangle  + \frac{1}{8\sqrt{2}} \langle X_A X_B Z_B X_B Z_B X_B \rangle  + \frac{1}{4\sqrt{2}} \langle X_A Z_A X_A Z_A X_A X_B Z_B X_B \rangle \\
& + \frac{1}{4\sqrt{2}} \langle X_A Z_A X_A X_B Z_B X_B Z_B X_B \rangle  - \frac{1}{8\sqrt{2}} \langle X_A Z_A X_A Z_A X_A X_B Z_B X_B Z_B X_B \rangle \Big] \end{aligned} \end{equation*}  \\ 

 \begin{align*} Z_A \otimes X_B \end{align*} & \begin{equation*} \begin{aligned} \frac{1}{2} \Big[ \frac{1}{2} & + \frac{1}{4\sqrt{2}} \langle Z_A X_B \rangle - \frac{1}{8\sqrt{2}} \langle X_A X_B \rangle + \frac{1}{8\sqrt{2}} \langle Z_A X_A Z_A X_B \rangle  - \frac{1}{2\sqrt{2}} \langle Z_A X_B Z_B X_B \rangle \\
& - \frac{1}{4\sqrt{2}} \langle X_A X_B Z_B X_B \rangle   - \frac{1}{8} \langle Z_A X_A Z_B X_B \rangle  - \frac{1}{8} \langle X_A Z_A X_B Z_B \rangle + \frac{1}{8} \langle X_A Z_A Z_B X_B \rangle \\
& + \frac{1}{8} \langle Z_A X_A X_B Z_B \rangle    + \frac{1}{4\sqrt{2}} \langle Z_A X_A Z_A X_B Z_B X_B \rangle  - \frac{1}{4\sqrt{2}} \langle Z_A X_B Z_B X_B Z_B X_B \rangle  \\
& + \frac{1}{8\sqrt{2}} \langle X_A X_B Z_B X_B Z_B X_B \rangle  - \frac{1}{8\sqrt{2}} \langle Z_A X_A Z_A X_B Z_B X_B Z_B X_B \rangle \Big]  \end{aligned} \end{equation*} \\ 

\end{tabular}
\end{ruledtabular}
\caption{Fidelity expressions for the DI setting. \label{table:difid}}
\end{table*}

\subsubsection{3. Constraints}

Our aim is to then determine a bound on the fidelity measures in Tables \ref{table:1sdifid} - \ref{table:difid}, given an $\epsilon$-near-maximal violation of the steering or CHSH inequalities. 

In the 1sDI setting, these inequalities can be written in the form of Alice's trusted operators $\sigma_X, \sigma_Z$ acting on assemblages that result from Bob measuring his untrusted operators $X, Z$. We write our inequalities in terms of assemblages on Alice's side as
\begin{align}
\text{ tr} [ \sigma_Z (2 \tau_{0|0} - \rho_A)  +  \sigma_X (2\tau_{0|1} - \rho_A) ] & = 2 - \epsilon, \\
\text{ tr} [\sqrt{2} \sigma_Z (2 \tau_{0|0} - \rho_A)   + \sqrt{2} \sigma_X (2\tau_{0|1} - \rho_A) ] & = 2\sqrt{2} - \epsilon, 
\end{align}
for the steering and CHSH inequalities, respectively.

In the DI case, where both parties' devices are untrusted, they must demonstrate nonlocality by violating the CHSH inequality.

\subsubsection{4. Finding a bound via SDP}

We then use an SDP to find the minimum value of our fidelity expression consistent with the amount of violation of the inequality. It is formulated as follows:
\begin{align}
 \text{minimise } & \text{tr} (P \Gamma)  \nonumber \\
 \text{such that } &  \Gamma \geq 0 \nonumber \\
 & \text{tr} (Q \Gamma) = w.
\end{align} 
Here, $\Gamma$ is the moment matrix from the NPA hierarchy for characterising quantum correlations \cite{Navascues2007, Navascues2008}. In the 1sDI setting, the rows are given by a set of operators that are some product of Bob's projectors $E_{\mathsf{b}|\mathsf{y}}$ written as $\{ \mathcal{E}_1, ... , \mathcal{E}_m \} $, the columns by the adjoints, and each element is given by $\Gamma_{k l} = \text{tr}_B (\mathcal{E}_l^\dag \mathcal{E}_k \ket{\psi} \bra{\psi})$. Thus, $\Gamma$ is an $m \times m$ matrix composed of assemblage elements. In the DI setting, the rows are given by a set of operators that are some combination of both Alice and Bob's operators $A_k, B_k$, which results in the elements of $\Gamma$ being expectation values. It has been proven in \cite{Supic, Bancal2015} that $\Gamma$ is positive semidefinite. 

Our constraint is then given by the amount of violation $w$ of a suitable inequality. Formulating our problem in terms of an SDP essentially amounts to finding a suitable form of $\Gamma$ that contains all the elements in our fidelity and inequality expressions, writing these expressions in terms of symmetric matrices $P, Q$ acting on $\Gamma$, and constraining the structure of $\Gamma$. 

For self-testing the measurements, this final step is the most time-consuming. For example, in the case of self-testing the $X \otimes X$ measurement in the DI scenario, to solve the SDP we require a $\Gamma$ of size $81 \times 81$, which has over $22,000$ constraints. To work around this problem, we have automated this process of generating $\Gamma$ and its constraints, using a combination of C++ and Python scripts. Our code takes as input the rows and columns of $\Gamma$. It then generates all the entries of $\Gamma$, determines which entries are equal to one another, and after some processing, outputs a list of all the unique constraints, in a format that can be directly entered into our SDP.

Once our SDP is set up, we solve it via CVX \cite{CVXResearchInc.2012, Grant2008}, giving the bounds in Theorem \ref{th:selftestingbounds}. An alternative form of our bounds is given in Table \ref{table:ourbounds}, for comparison with the literature.

\begin{table*}[t]
\centering
\begin{ruledtabular}
\begin{tabular}{llllll}
& \textit{Method} & \textit{Trust} & \textit{Inequality} & \textit{State bound} & \textit{State and measurements bound} \\   \hline 
Reichardt et al. \cite{Reichardt2013} & analytical & DI & CHSH & $320 \sqrt{\epsilon} $ &  \\ \\
McKague et al.  \cite{Mckaguea} & analytical & DI  & CHSH & $10.9\epsilon^{1/4} + 3.6 \sqrt{\epsilon} $ & $10.9\epsilon^{1/4} + 13.1 \sqrt{\epsilon} $ \\ \\
Kaniewski \cite{Kaniewski2016} & analytical & DI & CHSH & $1.18 \sqrt{\epsilon}$ & \\ \\
Bancal et al. \cite{Bancal2015} & numerical & DI  & CHSH & $1.48 \sqrt{\epsilon}$ &  \\ \\
Gheorghiu et al. \cite{Gheorghiu} & analytical & 1sDI  & steering & $2.8 \sqrt{\epsilon} + 0.5 \epsilon$ & $10.8 \sqrt{\epsilon} + 0.5 \epsilon$ \\ \\
\v Supi\'c et al. \cite{Supic} & numerical & 1sDI  & CHSH & $1.34 \sqrt{\epsilon}$ &  \\ \\
\multirow{3}{*}{\textit{This work}} & \multirow{3}{*}{numerical} & 1sDI  & steering & $1.59 \sqrt{\epsilon}$ & $2.49 \sqrt{\epsilon} $  \\ 
  & & 1sDI  & CHSH & $1.34 \sqrt{\epsilon} $ & $2.10 \sqrt{\epsilon}$ \\ 
 & & DI  & CHSH & $1.54\sqrt{\epsilon} $ & $2.72 \sqrt{\epsilon}$ 
\end{tabular}
\end{ruledtabular}
\caption{New and existing bounds on trace distance, given an $\epsilon$-near-maximal violation of the inequality. The state bound gives an upper bound on 
$\norm{ \Phi( \ket{\psi}) -   \ket{\overline{\psi}} \otimes \ket{anc} } \leq \sqrt{2f(\epsilon)} $, 
while the state and measurements bound gives an upper bound on 
$\norm{ \Phi(M_A \otimes N_B \ket{\psi}) - ( \overline{M}_A \otimes \overline{N}_B \ket{\overline{\psi}}) \otimes \ket{anc} }  \leq \sqrt{2f'(\epsilon)}$, 
where the relation between the bounds is calculated from Equation (39) of \cite{Bancal2015}.
In the 1sDI case, $M_A = \overline{M}_A$, since we trust Alice's measurement devices.
\label{table:ourbounds}}
\end{table*}
\end{proof}

\section{Appendix B: Proof of Theorem 2}
\label{sec:proofs}
\begin{manualtheorem}{2}
The fidelity of the entangled state used for teleportation in step 7 of Protocol 1 (up to local isometry) is bounded in the iid setting with probability at least $(1 - \epsilon^x)$ by
\begin{align}
F \geq 1 - \alpha \Big[ \frac{2\epsilon}{q} + \epsilon \Big],
\end{align}
and in the non-iid setting with probability at least $ (1-\epsilon^x)(1 - \sqrt{ \alpha [ \frac{2\epsilon}{q} + \frac{\epsilon}{2} + \frac{4q^2 x \epsilon \log{\frac{1}{\epsilon}} + 2 \epsilon^2}{8 q^2 x \log{\frac{1}{\epsilon}} +  \epsilon^2} ] })$ by
\begin{align}
F \geq 1 - \sqrt{ \alpha \Big[ \frac{2\epsilon}{q} + \frac{\epsilon}{2} + \frac{4q^2 x \epsilon \log{\frac{1}{\epsilon}} + 2 \epsilon^2}{8 q^2 x \log{\frac{1}{\epsilon}} +  \epsilon^2} \Big]  },
\end{align}
\label{th:fid}
where $\alpha = 1.26$.
\end{manualtheorem}
\begin{proof}
Our proof extends techniques from \cite{Gheorghiu, Hajdusek}. Let us assume Alice and Bob share $K$ Bell pairs, partitioned into $n$ pairs used for testing, and 1 pair used for teleportation. For each Bell pair, Alice and Bob measure either $X \otimes X$ or $Z \otimes Z$. Note that although we use an untested pair for teleportation, we will include this in our analysis as a `hypothetical measurement'. Since half of the pairs are measured in $X \otimes X$ and the other half in $Z \otimes Z$ in each test round, we have the number of tested pairs measured in either basis as $n_{XX} = n_{ZZ} = \frac{n}{2}$.

For now, we will assume that a pair measured in $X \otimes X$ is used for teleportation and the others for testing, and the pairs measured in $Z \otimes Z$ are only used for testing. We will see later that this assumption does not affect the results.
We also partition $K$ into $K_{XX}$ and $K_{ZZ}$, which are the total number of pairs measured in each respective basis. 
Thus we have $
K_{XX} = n_{XX}  + 1  = \frac{n + 2}{2}, \text{      }
K_{ZZ}  = n_{ZZ}   = \frac{n}{2} $.

The \textit{ideal correlation} for a pair measured in the basis $X \otimes X$ is given by $\mu_{XX} = \expval{\sigma_X \otimes \sigma_X}{\overline{\psi}}$, and similarly for $Z \otimes Z$. 
Note that $\mu_{XX} + \mu_{ZZ} = 2$. 

The \textit{measured correlation} in round $i$ is denoted by $\hat{C}_i = a_i b_i$, where $a_i, b_i$ are Alice and Bob's measurement outcomes which are either $\pm 1$. 
The average measured correlation can be written in terms of deviation from the ideal correlation, denoted by $\epsilon_{XX}, \epsilon_{ZZ}$ for the tested pairs, and $\epsilon'_{XX}$ for the pair $r$ used for teleportation, as
\begin{align}
\frac{1}{K_{XX}} \sum_{i = 1}^{K_{XX}} \hat{C}_i & =  \frac{1}{K_{XX}} \Big[ \sum_{i=1}^{n_{XX}}  \hat{C}_i  + \hat{C}_{r} \Big] \\ \label{eq:eqxx}
& =  \frac{1}{K_{XX}} \Big[ n_{XX} (\mu_{XX} - \epsilon_{XX})  + \mu_{XX} - \epsilon'_{XX}  \Big],  \\ 
\frac{1}{K_{ZZ}} \sum_{i = 1}^{K_{ZZ}} \hat{C}_i  & =  \frac{1}{K_{ZZ}} \Big[ \sum_{i=1}^{n_{ZZ}}  \hat{C}_i   \Big]  \\
& =  \frac{1}{K_{ZZ}} \Big[  n_{ZZ} (\mu_{ZZ} - \epsilon_{ZZ} ) \Big].
\label{eq:eqzz}
\end{align}

Using our tested pairs, the measured correlation showed a deviation $\epsilon$ from the ideal correlation,
giving $\epsilon = \epsilon_{XX} + \epsilon_{ZZ}$. 
We take $\epsilon_{XX} = \epsilon_{ZZ} = \frac{\epsilon}{2}$, but any other choice does not have a significant effect on our results. We will discuss the hypothetical deviation separately for the iid and non-iid cases.

The \textit{true correlation}, or expectation value, in a round $i$ where $X \otimes X$ is measured is denoted by $C_i = \text{tr} (X \otimes X \rho_i)$, and similarly for $Z \otimes Z$, where $\rho_i$ is the shared state in that round. 

Given the $\epsilon$-closeness between the ideal and measured correlation, we now compute the closeness between the ideal and true correlation over the Bell pairs measured in $X \otimes X$ and $Z \otimes Z$ separately. We can then apply our self-testing result from Theorem \ref{th:selftestingbounds} to bound the fidelity of the state. 
At first, we will assume iid, and then remove this assumption for the most general scenario.

\subsubsection{1. iid setting}

In the iid setting, the untrusted components are assumed to behave the same way in each round. 
This implies that the hypothetical correlation of the untested pair will be the same as a tested pair ($\epsilon'_{XX} = \epsilon_{XX}$). Note that this would be the same if we had used a pair measured in $Z \otimes Z$ for the teleportation.
Substituting in Equations (\ref{eq:eqxx}) and (\ref{eq:eqzz}), we get
\begin{align}
\frac{1}{K_{XX}} \sum_{i = 1}^{K_{XX}} \hat{C}_i & =  \mu_{XX} - \frac{\epsilon}{2}, \\
 \frac{1}{K_{ZZ}} \sum_{i = 1}^{K_{ZZ}} \hat{C}_i & = \mu_{ZZ} -  \frac{\epsilon}{2}.
\end{align}
Now, let us consider $X \otimes X$. We start by defining a variable $W_j 
= \sum_{i=1}^j (\hat{C}_i - C_i)$,  where  $j \in \{ 0, 1, ... , K_{XX} \}$.
We then use the Chernoff-Hoeffding bound \cite{Chernoff1952, Hoeffding1963} for approximating the expectation value of independent random variables. We have $g_i \leq \hat{C}_i \leq h_i$, where $g_i = -1, h_i = 1$. This gives 
\begin{align}
\text{Pr} (W_{K_{XX}} \geq \gamma) &  \leq   \exp(-\frac{2\gamma^2}{\sum_{i=1}^{K_{XX}} (h_i - g_i)^2}) \\
& \leq  \exp(-\frac{\gamma^2}{2 K_{XX}}).
\end{align}
We choose $\gamma = K_{XX} \epsilon$, which gives
\begin{align}
\text{Pr}  \Big(   \frac{1}{K_{XX}} & W_{K_{XX}}  \geq \epsilon \Big)   \nonumber \\
& = \text{Pr} \Big( \frac{1}{K_{XX}} \big[ \sum_{i = 1}^{K_{XX}} \hat{C}_i - \sum_{i = 1}^{K_{XX}} C_i \big] \geq \epsilon \Big) \\
& = \text{Pr} \Big( \mu_{XX} - \frac{\epsilon}{2} -  \frac{1}{K_{XX}}  \sum_{i = 1}^{K_{XX}} C_i    \geq \epsilon \Big) \\
& \leq  \exp(-\frac{1}{2} K_{XX} \epsilon^2). 
\end{align}
Thus, we have
\begin{align}
\text{Pr} \Big( \frac{1}{K_{XX}}   \sum_{i = 1}^{K_{XX}}  C_i  - \mu_{XX} \geq - \frac{3}{2}\epsilon \Big)  \geq  1 - \exp(-\frac{1}{2} K_{XX} \epsilon^2).
\label{eq:xx}
\end{align}
Similar calculations for $Z \otimes Z$ give us
\begin{align}
\text{Pr}  \Big( \frac{1}{K_{ZZ}}  \sum_{i = 1}^{K_{ZZ}} C_i & - \mu_{ZZ} \geq -\frac{3}{2} \epsilon \Big) \nonumber \\
&  \geq  1 - \exp(-\frac{1}{2} K_{ZZ} \epsilon^2).
\label{eq:zz}
\end{align}
Let $\rho_{avg} = \frac{1}{K} \sum_{i=1}^K \rho_i$ be the average state over all $K$ Bell pairs, including the one used for teleportation. Since the states in each round are identical in the iid setting, $\rho_{avg} = \rho_i$. We define the following averaged true correlations:
\begin{align}
C^{XX}   = \frac{1}{K_{XX}} \sum_{i=1}^{K_{XX}} C_i, \text{    } C^{ZZ}    = \frac{1}{K_{ZZ}} \sum_{i=1}^{K_{ZZ}} C_i.
\end{align}
We now rephrase Equations (\ref{eq:xx}) and (\ref{eq:zz}) as
\begin{align}
\text{Pr} \Big( C^{XX}  - \mu_{XX}   \geq -  \frac{3}{2}\epsilon \Big)  \geq 1 -   \exp(-\frac{1}{4} (n + 2)  \epsilon^2),  
\label{eq:xx2}
\end{align}
\begin{align}
\text{Pr} \Big( C^{ZZ} - \mu_{ZZ}   \geq  - \frac{3}{2} \epsilon \Big) 
 \geq 1 -   \exp(-\frac{1}{4} n  \epsilon^2).  
 \label{eq:zz2}
\end{align}
The true correlation for the averaged state when measured in the $X \otimes X$ basis is given by $C^{XX} = \text{tr} (X \otimes X \rho_{avg})$, and similarly for $Z \otimes Z$. Combining Equations (\ref{eq:xx2}) and (\ref{eq:zz2}), we get
\begin{align}
C^{XX} + C^{ZZ} - (\mu_{XX} + \mu_{ZZ}) & \geq - \frac{3}{2}\epsilon - \frac{3}{2} \epsilon \\
C^{XX} + C^{ZZ} & \geq 2 - 3\epsilon \\
& \geq 2 - \delta,
\end{align}
with probability $\geq \big[ 1 -   \exp(-\frac{1}{4} (n+2) \epsilon^2) \big] \big[ 1 -   \exp(-\frac{1}{4} n \epsilon^2) \big] \geq  1 -   \exp(-\frac{1}{4} n \epsilon^2) $, and where $\delta =  3\epsilon$.

If $\ket{\zeta}$ is a purification of $\rho_{avg}$, we get
\begin{align}
\expval{X \otimes X }{\zeta} + \expval{Z \otimes Z}{\zeta} \geq 2 - \delta,
\end{align}
with probability $\geq 1 -  \exp(-\frac{1}{4} n \epsilon^2)$. 
Our self-testing result in Theorem \ref{th:selftestingbounds} tells us that given such a correlation in expectation values, we have for $\alpha = 1.26$, 
\begin{align}
 F(\text{tr}_B [\Phi(\ket{\zeta}\bra{\zeta})], \ket{\overline{\psi}}) \geq 1 - \alpha \delta,
\end{align}
with probability $\geq 1 -  \exp(-\frac{1}{4} n \epsilon^2)$. 
This gives
\begin{align}
F(\text{tr}_B [\Phi(\rho_{avg})], \ket{\overline{\psi}}) \geq 1 - \alpha \delta,
\end{align}
with probability $\geq 1 -  \exp(-\frac{1}{4} n \epsilon^2)$.
If we set $n = \frac{4}{\epsilon^2} \log{\frac{1}{\epsilon}}$, we get $
F(\text{tr}_B [\Phi(\rho_{avg})], \ket{\overline{\psi}}) \geq 1 - \alpha \delta$,
with probability $\geq 1 - \epsilon$.
We rewrite this as follows, recalling that $\rho_{avg} = \rho_i$ in the iid setting: 
\begin{align}
F(\text{tr}_B [\Phi(\rho_{i})], \ket{\overline{\psi}})   \geq 1 - \alpha [3 \epsilon],
\end{align}
with probability $\geq 1 - \epsilon$.
The total number of pairs we need in our protocol is then $
K = \frac{4}{\epsilon^2} \log{\frac{1}{\epsilon}} + 1$.

Since our focus is on developing a protocol that can be implemented in the lab, we will now further optimise this analysis given our experimental limitations. If it is very difficult for us to observe a high violation of the steering inequality, but we are not limited in the number of pairs we can generate, then we can modify the above analysis as follows.
We now choose $\gamma = K_{XX} \frac{\epsilon}{q}$ in our Chernoff-Hoeffding bound, which gives
\begin{align}
\text{Pr} \Big( \frac{1}{K_{XX}} W_{K_{XX}} \geq \frac{\epsilon}{q} \Big) & \leq  \exp(-\frac{1}{2q^2} K_{XX} \epsilon^2).
\end{align}
Going through the calculations as above, we obtain
\begin{align}
\text{Pr} \Big( \frac{1}{K_{XX}}  \sum_{i = 1}^{K_{XX}} C_i & - \mu_{XX} \geq -\frac{\epsilon}{q} - \frac{\epsilon}{2} \Big) \nonumber \\
& \geq  1 - \exp(-\frac{1}{2q^2} K_{XX} \epsilon^2).
\end{align}
Similarly for $Z \otimes Z$, we get
\begin{align}
\text{Pr} \Big( \frac{1}{K_{ZZ}}  \sum_{i = 1}^{K_{ZZ}} C_i & - \mu_{ZZ} \geq -\frac{\epsilon}{q} - \frac{\epsilon}{2}  \Big) \nonumber \\
& \geq  1 - \exp(-\frac{1}{2q^2} K_{ZZ} \epsilon^2).
\end{align}
This gives us $C^{XX} + C^{ZZ} \geq 2 - \delta$ with probability  $\geq 1 - \exp(-\frac{1}{4q^2} n \epsilon^2)$,  where
$\delta = \frac{2\epsilon}{q} + \epsilon$. If we now set $n = \frac{4q^2 x}{\epsilon^2} \log{\frac{1}{\epsilon}}$, this gives us
\begin{align}
F(\text{tr}_B [\Phi(\rho_{i})], \ket{\overline{\psi}})  \geq 1 - \alpha \Big[ \frac{2\epsilon}{q} + \epsilon \Big],
\end{align}
with probability $\geq 1 - \epsilon^x$.
The total number of pairs we need in our protocol is then $
K = \frac{4q^2 x}{\epsilon^2} \log{\frac{1}{\epsilon}} + 1$. 
Note that as $q, x \rightarrow \infty$, the expression for fidelity reduces to the self-testing result of $F \geq 1 - \alpha \epsilon$ with probability 1. 

\subsubsection{2. Non-iid setting}

Now, we no longer assume the same behaviour throughout the rounds. 
We first determine the hypothetical correlation of the untested pair. We consider the worst case scenario i.e. the maximum possible error. Since $\mu_{XX} = 1$, this then gives $\epsilon'_{XX} = 2$.
Note that we get this same value if we had used the $Z \otimes Z$ pair for teleportation.
Substituting in Equations (\ref{eq:eqxx}) and (\ref{eq:eqzz}), we get
\begin{align}
\frac{1}{K_{XX}} \sum_{i = 1}^{K_{XX}} \hat{C}_i & =  \mu_{XX} - \Big[ \frac{8+n\epsilon}{2n+4} \Big], \\
 \frac{1}{K_{ZZ}} \sum_{i = 1}^{K_{ZZ}} \hat{C}_i & = \mu_{ZZ} -  \frac{\epsilon}{2}.
\end{align}

The true correlation $C_i$ now depends on the history of the measurements made, which we denote $H_i$, and so we can also write this as
\begin{align}
C_i  = \text{Pr} (a_i = b_i  | H_i) - \text{Pr} (a_i \neq b_i | H_i).
\end{align}
We have $\abs{W_{j+1} - W_j} \leq 2$, since $\hat{C}_i = \pm 1, -1 \leq C_i \leq 1$. The expectation value of $W_j$ is finite, and the conditional expected value $E(W_{j+1} | H_j) = W_j$. Thus, $\{ W_j \}$ is a martingale. We can use the Azuma-Hoeffding inequality \cite{Azuma1967, Hoeffding1963} here, which gives us for $\abs{W_{j+1} - W_j} \leq d_i$:
\begin{align}
\text{Pr} (W_{K_{XX}} \geq \gamma) &  \leq   \exp(-\frac{\gamma^2}{2 \sum_{i=1}^{K_{XX}} d_i^2}) \\
& \leq  \exp(-\frac{\gamma^2}{8 K_{XX}}).
\end{align}
We again choose $\gamma = K_{XX} \epsilon$, which gives
\begin{align}
\text{Pr} \Big( \frac{1}{K_{XX}} W_{K_{XX}} \geq \epsilon \Big) & \leq  \exp(-\frac{1}{8} K_{XX} \epsilon^2),
\end{align}
\begin{align}
\text{Pr} \Big( \frac{1}{K_{XX}}   \sum_{i = 1}^{K_{XX}} C_i & - \mu_{XX} \geq - \frac{3n\epsilon + 8 +  4 \epsilon}{2n+4} \Big) \nonumber \\
&  \geq  1 - \exp(-\frac{1}{8} K_{XX} \epsilon^2).
\end{align}
For $Z \otimes Z$, we get
\begin{align}
\text{Pr} \Big( \frac{1}{K_{ZZ}}  \sum_{i = 1}^{K_{ZZ}} C_i - \mu_{ZZ} \geq -\frac{3}{2} \epsilon \Big) & \geq  1 - \exp(-\frac{1}{8} K_{ZZ} \epsilon^2).
\end{align}
We substitute for $K_{XX}, K_{ZZ}$ and obtain $
\expval{X \otimes X }{\zeta} + \expval{Z \otimes Z}{\zeta} \geq 2 - \delta$
with probability $\geq 1 -  \exp(-\frac{1}{16} n \epsilon^2)$, where $\delta =  \frac{3n\epsilon + 4 + 5\epsilon}{n+2}$.  

If we now set $n = \frac{16}{\epsilon^2} \log{\frac{1}{\epsilon}}$ and use our self-testing result from Theorem \ref{th:selftestingbounds}, we get
\begin{align}
 F(\text{tr}_B [\Phi(\rho_{avg})], \ket{\overline{\psi}}) \geq 1 - \alpha \delta,
 \end{align}
with probability $\geq 1 - \epsilon$.

Note that here, $\rho_{avg}$ is not equal to $\rho_i$, since we do not assume the source behaves the same way in each round. Thus, in order to characterise the state in any round, we will extend Lemma 5 of \cite{Gheorghiu} as follows. 

\begin{manuallemma}{4}
Let $\rho_{avg} = \frac{1}{K} \sum_{i=1}^K \rho_i$. If, for some pure state $\ket{\Psi}$, we have that
\begin{align}
F(\rho_{avg}, \ket{\Psi}) \geq 1 - \eta,
\end{align}
then for a uniformly at random chosen $i \in \{ 1, ..., K\}$, with probability at least $1 - \sqrt{\eta}$ we can bound the fidelity by
\begin{align}
F(\rho_i, \ket{\Psi}) \geq 1 - \sqrt{\eta}.
\end{align}
\label{l:avg}
\end{manuallemma}
\begin{proof}
We can write $F(\rho_{avg}, \ket{\Psi})  \geq 1 - \eta$ as
\begin{align}
\frac{1}{K} \sum_{i=1}^K F(\rho_i, \ket{\Psi})  \geq 1 - \eta.
\end{align}
We wish to find a bound on any $F(\rho_i, \ket{\Psi})$, given this bound on the average $F(\rho_i, \ket{\Psi})$. Let us take $p$ to be the fraction of $i$'s such that $F(\rho_i, \ket{\Psi}) \leq 1 - \eta - \beta$, and $(1-p)$ to be the fraction of $i$'s such that $F(\rho_i, \ket{\Psi}) \geq 1 - \eta - \beta$. We take this to be the worst case scenario, so $F(\rho_i, \ket{\Psi}) = 1$ with probability $(1-p)$. We can then  write
\begin{align}
\frac{1}{K} \sum_{i=1}^K F(\rho_i, \ket{\Psi}) & \leq p \times (1 - \eta - \beta) + (1-p) \times 1 \\
1 - \eta & \leq p \times (1 - \eta - \beta) + 1 - p.
\end{align}
This gives $p \leq \frac{\eta}{\eta + \beta}$, and $1 - p \geq \frac{\beta}{\eta+ \beta}$.
Thus, if we choose a random $i$, then with probability $\geq \frac{\beta}{\eta + \beta}$ its fidelity can be lower-bounded by $F(\rho_i, \ket{\Psi}) \geq 1 - \eta - \beta$. We choose $\beta = \sqrt{\eta} - \eta$, which gives
\begin{align}
F(\rho_i, \ket{\Psi}) \geq 1 - \sqrt{\eta}, 
\end{align}
with probability $\geq 1 - \sqrt{\eta}$.
\end{proof}

\begin{figure*}[ht]
\centering
\includegraphics[trim = 0cm 12cm 0cm 0.5cm, width=\textwidth]{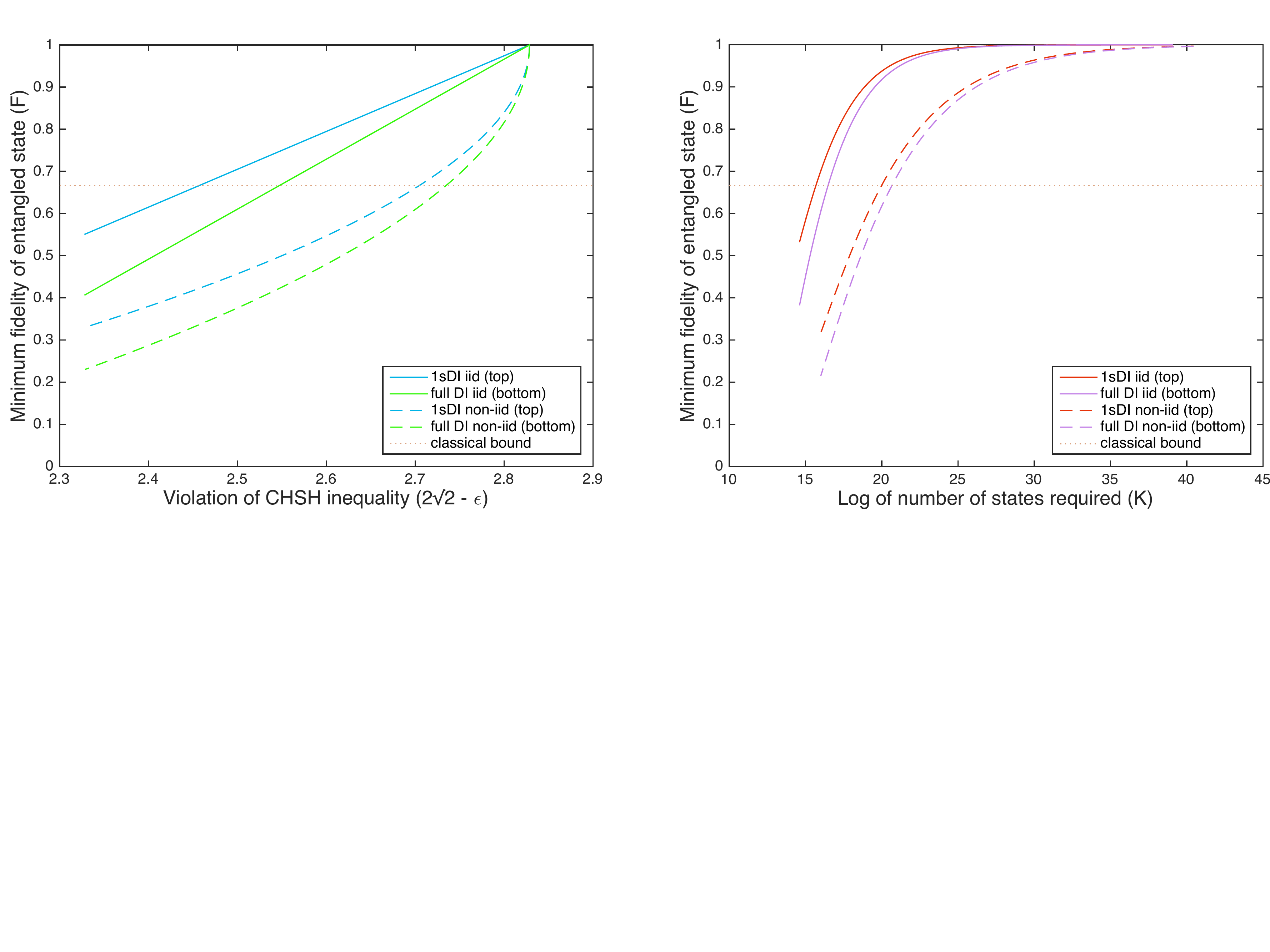}
\caption{\small Certifying fidelities of the entangled state for 1sDI and DI authenticated teleportation, in the iid (solid line) and non-iid (dashed line) settings, requires (a) a CHSH inequality violation of $2\sqrt{2}-\epsilon$, and (b) at least $K$ states. The classical bound (dotted line) corresponds to a fidelity of $\frac{2}{3}$.}
\label{fig:di}
\end{figure*}

Applying Lemma \ref{l:avg} to our scenario, we obtain
\begin{align}
F(\text{tr}_B [\Phi(\rho_i)], \ket{\overline{\psi}}) \geq 1 - \sqrt{\alpha \delta},
\end{align}
with probability $\geq (1-\epsilon)(1 - \sqrt{ \alpha \delta })$.

Writing $\delta$ in terms of $\epsilon$, we have our final result that the state $\rho_i$ in any round, including the one used for teleportation, is such that 
\begin{align}
F(\text{tr}_B [\Phi( \rho_i )], \ket{\overline{\psi}}) 
 &  \geq 1  - \sqrt{ \alpha \Big[  \frac{48\epsilon\log{\frac{1}{\epsilon}} + 4\epsilon^2 + 5\epsilon^3}{16\log{\frac{1}{\epsilon}} + 2\epsilon^2} \Big] },
 \end{align}
 with probability $ \geq (1-\epsilon)(1  - \sqrt{ \alpha \Big[  \frac{48\epsilon\log{\frac{1}{\epsilon}} + 4\epsilon^2 + 5\epsilon^3}{16\log{\frac{1}{\epsilon}} + 2\epsilon^2} \Big] }).$
 
 The total number of pairs we need in our protocol is then $K = \frac{16}{\epsilon^2} \log{\frac{1}{\epsilon}} + 1$.

If we wish to optimise this further given our experimental limitations, we can do a similar modification to the above analysis as in the iid setting, which gives us
\begin{align}
F( \text{tr}_B [\Phi(\rho_i)],  \ket{\overline{\psi}} )
& \geq 1  - \sqrt{ \alpha \Big[ \frac{2\epsilon}{q} + \frac{\epsilon}{2} + \frac{4q^2 x \epsilon \log{\frac{1}{\epsilon}} + 2 \epsilon^2}{8 q^2 x \log{\frac{1}{\epsilon}} +  \epsilon^2} \Big] },
\end{align}
with probability 
$\geq (1-\epsilon^x)( 1  - \sqrt{ \alpha \Big[ \frac{2\epsilon}{q} + \frac{\epsilon}{2} + \frac{4q^2 x \epsilon \log{\frac{1}{\epsilon}} + 2 \epsilon^2}{8 q^2 x \log{\frac{1}{\epsilon}} +  \epsilon^2} \Big] } ). 
$

The total number of pairs used in the protocol is $K = \frac{16q^2 x}{\epsilon^2} \log{\frac{1}{\epsilon}} + 1$.
\end{proof}

\section{Appendix C: Results using the CHSH inequality}
\label{sec:reschsh}

In the 1sDI setting of Protocol 1, the parties can also use the CHSH inequality instead of the steering inequality. We then state the analogous statement to Theorem \ref{th:fid} and Corollary 3 here, for which the proof follows by the same method as above.

\begin{manualtheorem}{5}
Alice and Bob perform Protocol 1 but using the CHSH inequality. 
In the iid setting, if they share $\frac{8q^2 x}{\epsilon^2} \log{\frac{1}{\epsilon}} + 1$ pairs, then the fidelity of teleportation is bounded with probability at least $(1-\epsilon^x)$ by
\begin{align}
F
 \geq 1 - \alpha \Big[ \frac{4\epsilon}{q} + \epsilon \Big],
\end{align}
and in the non-iid setting, if they share $\frac{32q^2x}{\epsilon^2} \log{\frac{1}{\epsilon}} +1$ pairs, then the fidelity of teleportation is bounded with probability at least
$ (1-\epsilon^x)(1 - \sqrt{ \alpha \Big[ \frac{4\epsilon}{q} + \frac{3\epsilon}{4} + \frac{4q^2 x \epsilon \log{\frac{1}{\epsilon}} + (2 +  \sqrt{2}) \epsilon^2}{16 q^2 x \log{\frac{1}{\epsilon}} + 2 \epsilon^2} \Big]})$ by
\begin{align}
F
& \geq 1 - \sqrt{ \alpha \Big[ \frac{4\epsilon}{q} + \frac{3\epsilon}{4} + \frac{4q^2 x \epsilon \log{\frac{1}{\epsilon}} + (2 +  \sqrt{2}) \epsilon^2}{16 q^2 x \log{\frac{1}{\epsilon}} + 2 \epsilon^2} \Big]},
\end{align}
where $\alpha = 0.90$. 
\label{th:chsh1sdi}
\end{manualtheorem}

Note that we do not use the fact that one party's devices may be trusted apart from when we apply the SDP result. Thus, in the DI setting, Theorem \ref{th:chsh1sdi} holds with $\alpha = 1.19$, using the same number of pairs as specified there. This can be used to do a fully device-independent test of the Bell pair, useful for various applications. 

In Figure \ref{fig:di}, we compare the parameters required in the 1sDI and DI trust scenarios, under the iid and non-iid assumptions. In the 1sDI case, from Corollary 3, this fidelity bound holds for the fidelity of teleportation. In the DI case, one must also consider Alice's Bell state measurement in order to make such a statement; thus, this remains a bound on the entangled state.

\end{document}